%% file: selection-7-of-16.tex
\documentclass[11pt]{article}

\usepackage[utf8]{inputenc}
\usepackage[T1]{fontenc}
\usepackage{lmodern}
\usepackage[margin=1in]{geometry}
\usepackage{amsmath,amssymb,amsthm}
\usepackage{booktabs}
\usepackage{placeins}
\usepackage{enumitem}
\usepackage{microtype}
\usepackage{parskip}
\usepackage{xcolor}
\usepackage{tikz}
\usepackage{listings}
\usepackage[hidelinks]{hyperref}
\usepackage{orcidlink}
\newcommand{\doi}[1]{doi:\href{https://doi.org/#1}{#1}}

\newtheorem{theorem}{Theorem}
\newtheorem{lemma}[theorem]{Lemma}
\newtheorem{corollary}[theorem]{Corollary}
\newtheorem{proposition}[theorem]{Proposition}
\theoremstyle{remark}
\newtheorem{remark}[theorem]{Remark}

\newcommand{\Down}{\mathcal{D}}
\newcommand{\FG}{\mathrm{FG}}
\newcommand{\FGd}{\mathrm{FG}^{\delta}}
\newcommand{\KEY}{\mathrm{KEY}}

\title{Selecting the 7th Smallest of 16 Elements Takes Exactly 28 Comparisons}
\author{Jonathan Cremer\,\orcidlink{0009-0004-3018-1284}}
\date{}

\begin{document}
\maketitle

\begin{abstract}
\noindent
Let $V_i(n)$ denote the minimum worst-case number of comparisons needed to
find the $i$-th smallest of $n$ distinct elements. We construct algorithms
showing $V_7(16)\le28$ and $V_8(16)\le29$, improving the previously reported
upper bounds of 30 and 31. Together with known lower
bounds, these results establish $V_7(16)=28$ and leave
$V_8(16)\in\{28,29\}$. For 17 elements, we establish $V_5(17)=27$ and
improve the upper bounds for ranks 6 through 9 to 29, 30, 31, and 31,
respectively. Every new upper bound has an explicit decision-graph
certificate verified by a static checker developed independently of the
search. We also derive a lower bound for selection from a partially ordered
state by combining the Fussenegger--Gabow weight with its order-reversed
counterpart. All new upper bounds and this lower bound are formally verified
in Lean~4. For
$V_8(16)$, we show that a lower bound based solely on the total number of
decision-tree leaves cannot establish $V_8(16)\ge29$. The algorithms were
found using fixed-budget search, information-guided comparison ordering,
and proof-number search. Proof-number search was decisive in finding the
28-comparison algorithm; in our experiments, refuting a comparison budget
remained substantially more expensive than finding an algorithm within it.
\end{abstract}

\section{Introduction}

Selection asks for the element of a given rank using as few pairwise
comparisons as possible. Write $V_i(n)$ for the smallest worst-case number
of comparisons needed to find the $i$-th smallest of $n$ distinct elements.
Reversing the order shows $V_i(n)=V_{n+1-i}(n)$. Classical results include
$V_1(n)=n-1$ and $V_2(n)=n-2+\lceil\log_2 n\rceil$ (Kislitsyn;
see~\cite{knuth}), and the general upper bound of Hadian and
Sobel~\cite{hadian-sobel},
\[
  V_i(n)\le n-i+(i-1)\lceil\log_2(n+2-i)\rceil .
\]
Asymptotically, selection takes linear time~\cite{bfprt}, but even the
leading constant for the median is open~\cite{bent-john,dz-lower,dz-upper}.

Small exact values have been determined by computer search, starting with
Gasarch, Kelly and Pugh~\cite{gkp} and Oksanen~\cite{oksanen-web,oksanen}.
D\"orrer et al.~\cite{sea25} settled all entries for $n\le15$ and the first
six ranks for $n=16$. The two remaining 16-element entries had lower bound
28; with the upper bounds reported in~\cite[Table~3]{gkp}, the ranges were
$28\le V_7(16)\le30$ and $28\le V_8(16)\le31$.

\begin{theorem}\label{thm:main7}
$V_7(16)=V_{10}(16)=28$.
\end{theorem}

\begin{theorem}\label{thm:main8}
$V_8(16)=V_9(16)\le29$. Hence $V_8(16)\in\{28,29\}$.
\end{theorem}

\begin{theorem}\label{thm:main17}
$V_5(17)=V_{13}(17)=27$. Moreover $V_6(17)\le29$, $V_7(17)\le30$,
$V_8(17)\le31$, and $V_9(17)\le31$.
\end{theorem}

Our contribution to all three theorems is the upper bounds. The lower bound
of Theorem~\ref{thm:main7} is the entry 28--33 of~\cite[Table~3]{sea25},
obtained there by computer search; we use it without rederiving it. The
lower bounds $V_8(16)\ge28$ and $V_5(17)\ge27$ follow from the published
values $V_8(15)=27$ and $V_5(16)=26$~\cite{sea25} and the transfer lemma of
Section~\ref{sec:prelim}. Tables~\ref{tab:row16} and~\ref{tab:row17} show
the previous and new ranges. The table in~\cite{sea25} lists the
Hadian--Sobel values 33 and 36 for the two open 16-element entries; the
upper bounds 30 and 31 of~\cite{gkp} are smaller and are the baseline we
improve.

\begin{table}[htbp]
\centering
\small
\begin{tabular}{lcccccccc}
\toprule
$i$ & 1 & 2 & 3 & 4 & 5 & 6 & 7 & 8 \\
\midrule
before~\cite{sea25,gkp} & 15 & 18 & 21 & 24 & 26 & 27 & 28--30 & 28--31 \\
this paper & 15 & 18 & 21 & 24 & 26 & 27 & \textbf{28} & \textbf{28--29} \\
\bottomrule
\end{tabular}
\caption{The row $n=16$ of the selection table.}
\label{tab:row16}
\end{table}

\begin{table}[htbp]
\centering
\small
\begin{tabular}{lccccccccc}
\toprule
$i$ & 1 & 2 & 3 & 4 & 5 & 6 & 7 & 8 & 9 \\
\midrule
before & 16 & 20 & 22 & 25 & 27--28 & 28--31 & 29--33 & 29--34 & 29--35 \\
this paper & 16 & 20 & 22 & 25 & \textbf{27} & 28--\textbf{29} & 29--\textbf{30} & 29--\textbf{31} & 29--\textbf{31} \\
\bottomrule
\end{tabular}
\caption{The row $n=17$. Lower bounds: classical formulas for $i\le2$,
otherwise the row $n=16$ of~\cite{sea25} with the transfer lemma. Previous
upper bounds: Hadian--Sobel for $i\le4$ and~\cite[Table~3]{gkp} for
$i\ge5$.}
\label{tab:row17}
\end{table}

\paragraph{Contributions and evidence.}
\begin{enumerate}[leftmargin=*,itemsep=2pt]
  \item Seven new upper bounds, each witnessed by an explicit certificate
  (Section~\ref{sec:cert}). A static checker that shares no code with the
  search accepts every certificate, and the Lean~4 kernel checks a formal
  proof of every bound.
  \item A description of the 15-comparison opening shared by several of
  the new algorithms (Section~\ref{sec:structure}).
  \item A lower bound for selection from a partial order, combining the
  weight of Fussenegger and Gabow~\cite{fg} with its order-reversed
  counterpart (Section~\ref{sec:lb}), with a formal proof in Lean~4 for all
  finite partial orders.
  \item A search method that made the constructions practical
  (Section~\ref{sec:search}).
\end{enumerate}
Statements that rest only on search runs, such as refutations under a fixed
opening, are labeled as such and are not used in any theorem.

\section{Preliminaries}\label{sec:prelim}

A selection problem is a strict partial order $P$ on a set $\Omega$ of $n$
elements, recording the comparisons known so far and their transitive
consequences, together with a rank $i$. Write $V_i(P)$ for the minimum
worst-case number of further comparisons. We write $E_n$ for the empty
order on $n$ elements (the antichain), in which nothing is known. Since an
algorithm starts in this state, $V_i(n)=V_i(E_n)$ by definition.
The inputs are totally ordered, but the algorithm knows only $P$. Two
elements whose order is not yet known, that is, not implied by $P$, are
called incomparable, written $a\parallel b$; only such pairs are worth
comparing. $P+ab$ denotes the transitive closure of $P$ with $a<b$ added.
For an unsolved problem
\[
  V_i(P)=1+\min_{a\parallel b}\max\{V_i(P+ab),V_i(P+ba)\},
\]
and $V_i(P)=0$ once some element is known to have exactly $i-1$ elements
below and $n-i$ above it.

An element with at least $i$ known predecessors is too large to be the
answer, and one with at least $n-i+1$ known successors is too small.
Deleting such an element does not change the optimal cost~\cite{sea25};
deleting a too-small element lowers the rank by one. Problems that are
isomorphic, or dual under order reversal with rank $n+1-i$, have the same
cost.

\paragraph{Transfer lemma~\cite[Lemma~1]{sea25}.}
For $1\le i\le n$,
\[
 V_i(n+1)\ge V_i(n)+1,
 \qquad
 V_{i+1}(n+1)\ge V_i(n)+1.
\]
For the first inequality, run an algorithm for rank $i$ of $n+1$ elements
on the $n$ given elements and an artificial element larger than all of them,
answering comparisons with that element without cost. It returns the
$i$-th smallest given element. On every path the artificial element must be
compared at least once, since it must end up known to lie above the answer,
so every path saves at least one comparison. An artificial smallest element
gives the second inequality.

\section{Certificates and their verification}\label{sec:cert}

\subsection{Certificates}

We use the standard decision-tree model~\cite[Sec.~5.3]{knuth}. For fixed
$n$, a deterministic algorithm that accesses the elements only through
comparisons corresponds to a binary tree: each internal node compares two
elements and branches on the outcome, and each leaf returns an element. An
input follows a single root-to-leaf path, so the worst-case number of
comparisons is the depth of the tree, and $V_i(n)$ is the minimum depth of
a tree that always returns the $i$-th smallest element. An upper bound
$V_i(n)\le k$ is therefore proved by exhibiting one such tree of depth at
most $k$; a lower bound must hold for every such tree. Merging identical
subtrees turns the tree into a directed acyclic graph. Our certificates are such graphs, written
as Rust programs in the format of the search code of~\cite{sea25}. Each
function solves a remaining subproblem: it may compare two arguments,
discard arguments that cannot be the answer, and call another function on
the remaining ones. A doc comment declares the size, the 0-based rank, and
known relations. For example, this function from the 28-comparison
certificate finds the median of three elements when $b<c$ is known:
\begin{lstlisting}
/// n = 3, i = 1, b < c
fn select_12([a, b, c]: [usize; 3]) -> usize {
    if a < b { select_1([b]) } else { select_2([a, c]) }
}
\end{lstlisting}
If $a<b$, then $a<b<c$ and $b$ is returned. Otherwise $b$ is below both
$a$ and $c$, so it is discarded and the minimum of $a$ and $c$ is found with
one more comparison. The root function \texttt{select\_0} receives all
inputs with no known relations.

\subsection{The static checker and why it is sound}\label{sec:static}

The checker was written from the mathematical definitions without access to
the search code, and accepts only a narrow grammar. It never uses the
relations stated in comments. Instead it propagates knowledge from the root:
each comparison outcome is added and closed transitively, and arguments are
mapped to the callee's names. A function called from several places gets as
its precondition $P_f$ the relations that hold at \emph{every} call site;
this can only make acceptance harder. Under $P_f$ and each branch outcome
it requires:
\begin{enumerate}[leftmargin=*,itemsep=2pt]
  \item compared elements are incomparable before the comparison (this
  rejects wasted comparisons; soundness does not need it);
  \item each call passes the declared number of pairwise distinct arguments;
  \item each discarded element is certainly below the requested 0-based rank
  $r$ among $m$ elements (at least $m-r$ known successors) or certainly above
  it (at least $r+1$ known predecessors);
  \item the callee's declared rank equals $r$ minus the number of discarded
  elements below;
  \item a directly returned element has exactly $r$ known predecessors and
  $m-1-r$ known successors.
\end{enumerate}
It also checks that every function is reachable, the root is never called,
and the call graph is acyclic.

\paragraph{Soundness.}
Call $f$ correct if for all distinct inputs whose order
extends $P_f$ it terminates and returns the element of its declared rank.
By induction in reverse topological order, which exists because the call
graph is acyclic: an input of $f$ follows one checked branch; a returned
element is correct by rule~5; for a call, the arguments are distinct inputs
(rule~2) that satisfy the callee's precondition, since that precondition
was computed from what every caller knows; and rules~3 and~4 show that the
callee's answer is the caller's answer. Termination follows from
acyclicity. The root has the empty precondition, so the program is correct
on all inputs.

\paragraph{Comparison bound.}
Each input follows a single path from the root to a return and makes only
the comparisons on that path. The worst case is therefore the largest
number of comparisons on any path, not the size of the graph. The checker
computes it backward: a return costs $0$, and a comparison costs one more
than the larger cost of its two branches. Paths that no input can follow
only make this bound larger, so it remains a valid upper bound.

\paragraph{Tests.}
The checker accepts all 70 optimal algorithms published
with~\cite{sea25} with exactly their known costs, and rejected every wrong
or ill-formed program among 12,046 generated mutants. An independently
written exact oracle, following every path without intersecting
preconditions, confirms each of our certificates and finds that its exact
worst case equals the checker's bound. Each certificate also passes $2\cdot10^7$
compiled executions on pseudorandom permutations.

\subsection{Formal proofs in Lean}

A Lean~4 development formalizes the same argument. It uses only Lean's
core library, so every definition the theorems depend on is contained in
our own short specification. It defines programs of this kind, their
execution with a comparison counter, and the predicate $\mathrm{SelectsWithin}(Q,n,r,w)$:
on every list of $n$ distinct natural numbers, program $Q$ returns an input
with exactly $r$ smaller inputs, using at most $w$ comparisons. Programs
access inputs only by comparing them and passing them on, so they are
comparison algorithms, and replacing $n$ distinct elements of any totally
ordered set by their ranks $0,\dots,n-1$ leaves every execution unchanged.
The development proves that acceptance by a checker implies the predicate,
and the kernel evaluates that checker on the certificate data, giving
\[
 \exists Q\ \mathrm{SelectsWithin}(Q,n,r,w)\quad\text{for}\quad
 (n,r,w)\in\left\{\begin{array}{@{}l@{}}
   (16,6,28),\ (16,7,29),\ (17,4,27),\ (17,5,29),\\
   (17,6,30),\ (17,7,31),\ (17,8,31)
 \end{array}\right\}.
\]
The proofs use only the standard axioms \texttt{propext},
\texttt{Classical.choice} and \texttt{Quot.sound}, not
\texttt{native\_decide}, and they replay in the independent kernel checker
\texttt{leanchecker}. The bounds do not depend on the tool that translates
the Rust certificates into Lean: the theorems concern the translated
programs themselves, and each is a valid witness. Only the further claim
that these programs are the archived Rust files depends on the
translation. We cross-checked it against an independent parser and
interpreter. The lower-bound theory of Section~\ref{sec:lb} is formalized in
a separate development described there.

\begin{table}[htbp]
\centering
\small
\begin{tabular}{lrrl}
\toprule
Bound & Functions & Search time & SHA-256 prefix \\
\midrule
$V_7(16)\le28$ & 329 & 5\,s & \texttt{eb8f064c12bb} \\
$V_8(16)\le29$ & 2937 & 5\,min & \texttt{2e9615ae1d9b} \\
$V_5(17)\le27$ & 146 & 103\,s & \texttt{34e26f201b68} \\
$V_6(17)\le29$ & 266 & 89\,s & \texttt{ae6d6b42e1c4} \\
$V_7(17)\le30$ & 977 & 8\,min & \texttt{33718071f64a} \\
$V_8(17)\le31$ & 1077 & 19\,min & \texttt{5925b96967a8} \\
$V_9(17)\le31$ & 1493 & 63\,min & \texttt{9ce649f22b2d} \\
\bottomrule
\end{tabular}
\caption{The certificates used in the Lean proofs, found by proof-number
search (single thread, AMD Ryzen~7 5800X). Further certificates for the
16-element bounds, with 365 and 252 functions, are archived.}
\label{tab:certs}
\end{table}

\section{The 28-comparison algorithm}\label{sec:structure}

The 329-function certificate opens with 15 comparisons that reach the same
partial order on every input, followed by an irregular endgame of at most 13
comparisons. Calling the smaller element of a comparison its winner: compare
eight disjoint pairs; compare the winners of pairs to form four
four-element tournaments; then run a tournament among the four losers of
these finals, the \emph{runner-ups}. Each four-element group then has
winner $w$, its first opponent $b$, runner-up $r$, and the runner-up's first
opponent $d$, with $w<b$, $w<r$, $r<d$. Comparing the runner-ups of two
such groups gives $R_8$, and comparing the smaller runner-ups of two copies
of $R_8$ gives $R_{16}$ (Figure~\ref{fig:v7structure}). Swapping an outcome
only renames components, so every input reaches $R_{16}$ after
$8+4+3=15$ comparisons.

\input{fig_v7_structure}

For example, the 16th comparison is $r_B$ against $b_C$. If $r_B<b_C$, then
$w_A$ lies below the ten elements $r_A,d_A,b_A,r_B,d_B,b_C,r_C,d_C,r_D,d_D$,
so it cannot be the 7th smallest; it is discarded, leaving the 6th smallest
of 15 with 12 comparisons. We found no short set of rules describing the
endgame. The second 28-comparison certificate and the 252-function
29-comparison certificate for the 8th smallest also pass through $R_{16}$.
The opening is not optimal for every rank: the bound of
Section~\ref{sec:lb} excludes ranks 1--3, and for the 8th smallest our
search finds no 28-comparison completion (a search-level result), so
$R_{16}$ is optimal there only if $V_8(16)=29$.

\FloatBarrier
\section{Lower bounds from a partially ordered state}\label{sec:lb}

This section bounds from below the number of comparisons $V_i(P)$ that are
still needed once some comparisons are known. Such bounds let the search
discard states whose remaining budget is provably too small. We combine two
classical weights into a new bound on the number of leaves of every correct
decision tree (Theorem~\ref{thm:phi}), recall a bound that follows a single
input (Theorem~\ref{thm:key}), show that the combined weight is the best in
a natural family (Proposition~\ref{prop:phiopt}), and explain why no leaf
count alone can settle $V_8(16)$ (Remark~\ref{rem:leafcount}).

A binary tree of depth $d$ has at most $2^d$ leaves, so a lower
bound $N$ on the number of leaves of every correct tree gives
$V_i(P)\ge\lceil\log_2 N\rceil$.

A \emph{compatible solution} is a triple $(L,e,U)$ with $|L|=i-1$ such that
some total order extending $P$ lists $L$, then $e$, then $U$; let $C(P)$ be
their set. At a correct leaf the known relations determine the sets below
and above the returned element $e$: otherwise $e$ could be placed directly
after its known predecessors in one consistent total order and directly
before its known successors in another, giving different ranks. Hence
distinct solutions need distinct leaves and
$V_i(P)\ge\lceil\log_2|C(P)|\rceil$~\cite[Thm.~6]{sea25}. This is weak:
finding the minimum of four elements has four solutions but needs three
comparisons, because each non-minimum must lose once.

\paragraph{Weights.}
A \emph{down-set} contains all known predecessors of its members; let
$\Down_j(P)$ be the down-sets of size $j$. For $L\in\Down_{i-1}(P)$, the
possible answers are the minimal elements of $\Omega\setminus L$; let $b(L)$
be their number minus one. For $A\in\Down_i(P)$, the possible answers are
the maximal elements of $A$; let $a(A)$ be their number minus one. A
solution corresponds to $L$ and $A=L\cup\{e\}$, so
$|C(P)|=\sum_L(b(L)+1)=\sum_A(a(A)+1)$. Define
\[
  \FG(P)=\sum_{L\in\Down_{i-1}(P)}2^{b(L)},\qquad
  \FGd(P)=\sum_{A\in\Down_i(P)}2^{a(A)},\qquad
  \Phi^*(P)=\FG(P)+\FGd(P)-|C(P)|.
\]
Fussenegger and Gabow~\cite[Thm.~2(b)]{fg} give a weight of this kind for
arbitrary partial orders (stated for the largest elements); after reversing
the order, their weight is at most $\FG(P)$, with equality when no comparisons are known.
The code of~\cite{sea25} computes $\max(\FG,\FGd)$. The new element is the
combination $\Phi^*$. For the minimum of four elements, $\FG=8$, $\FGd=4$,
$|C|=4$, and $\Phi^*=8$ gives the correct bound 3.

The key step is to follow how a single comparison changes the weight of
one fixed lower set $L$. After the known order has grown from $P$ to $Q$,
the set $L$ contributes
\[
  \psi_L(Q)=\begin{cases}
    2^{b_Q(L)} & \text{if $L$ is still a down-set of $Q$,}\\
    0 & \text{otherwise (then $L$ can no longer be the set below the answer),}
  \end{cases}
\]
where $b_Q(L)$ is $b(L)$ computed in $Q$. The lemma says that the weight of
$L$ never shrinks in total over the two outcomes.

\begin{lemma}\label{lem:world}
Let $x$ and $y$ be incomparable in $P$, let $P_1=P+xy$ and $P_2=P+yx$ be
the two possible outcomes of comparing them, and fix
$L\in\Down_{i-1}(P)$.
\begin{enumerate}[leftmargin=*,itemsep=2pt]
  \item If both $x$ and $y$ lie in $L$, the weight is unchanged in both
  outcomes: $\psi_L(P_1)=\psi_L(P_2)=\psi_L(P)$.
  \item If exactly one of them lies in $L$, one outcome keeps the weight
  $\psi_L(P)$ and the other sets it to $0$.
  \item If neither lies in $L$, the two outcomes together keep at least the
  weight: $\psi_L(P_1)+\psi_L(P_2)\ge\psi_L(P)$.
\end{enumerate}
The same holds for an initial set $A\in\Down_i(P)$ with weight $2^{a(A)}$,
with inside and outside exchanged: the weight is unchanged if $x,y\notin A$,
one outcome keeps it if exactly one of them lies in $A$, and the two
outcomes together keep at least the weight if $x,y\in A$.
\end{lemma}

\begin{proof}
Adding $x<y$ creates only relations $u<v$ with $u\le x$ and $y\le v$. Let
$M=\Omega\setminus L$. If $x\in L$, every $u\le x$ lies in $L$, so no
relation inside $M$ and none from $M$ into $L$ is created; this gives the
first case and the surviving outcome of the second. The outcome placing an
element of $M$ below one of $L$ destroys $L$ as a down-set. If $x,y\in M$,
then $L$ stays a down-set in both outcomes, and a minimal $v$ of $M$ can
gain a predecessor in $M$ only through $y\le v$, which forces $v=y$. So
$b(L)$ drops by at most one in each child, and
$\psi_L(P_1)+\psi_L(P_2)\ge2\cdot2^{b(L)-1}$. Order reversal gives the
statement for $A$.
\end{proof}

\begin{theorem}\label{thm:phi}
Every correct decision tree for $(P,i)$ has at least $\Phi^*(P)$ leaves.
Hence $V_i(P)\ge\lceil\log_2\Phi^*(P)\rceil$.
\end{theorem}

\begin{proof}
Removing comparisons with known outcomes does not increase the number of
leaves, and at a correct leaf the problem is solved, where
$\FG=\FGd=|C|=1$ and $\Phi^*=1$. It therefore suffices to show
$\Phi^*(P)\le\Phi^*(P_1)+\Phi^*(P_2)$ for incomparable $x,y$, and to sum
this inequality over the tree. Write $\Delta F=F(P_1)+F(P_2)-F(P)$.
Every compatible solution survives in at least one child, the children's
solutions are solutions of $P$, and a solution survives in both exactly
when $x$ and $y$ lie on the same side of its answer. Hence
\[
 \Delta|C|=\sum_{L\in\Down_{i-1}(P),\ x,y\in L}(b(L)+1)
          +\sum_{A\in\Down_i(P),\ x,y\notin A}(a(A)+1).
\]
Down-sets of the children are down-sets of $P$. By
Lemma~\ref{lem:world}, every $L$ contributes at least $0$ to $\Delta\FG$,
and $L$ with $x,y\in L$ contributes exactly $2^{b(L)}$; likewise for
$\Delta\FGd$ and $A$ with $x,y\notin A$. Therefore
\[
 \Delta\Phi^*\ge\sum_{x,y\in L}\bigl(2^{b(L)}-b(L)-1\bigr)
   +\sum_{x,y\notin A}\bigl(2^{a(A)}-a(A)-1\bigr)\ge0,
\]
since $2^t\ge t+1$.
\end{proof}

A second bound follows a single input. For a solution $s=(L,e,U)$ let
$a_s=a(L\cup\{e\})$, $b_s=b(L)$, and $\KEY(P)=\max_{s\in C(P)}(a_s+b_s)$.

\begin{theorem}\label{thm:key}
$V_i(P)\ge\KEY(P)$.
\end{theorem}

\begin{proof}
This is the classical argument about crucial
comparisons~\cite[Sec.~5.3.3]{knuth}, \cite{fg}, applied along one total
order. Answer all comparisons according to a total order extending $P$ that
lists $L$, then $e$, then $U$; let $A=L\cup\{e\}$ and $M=\Omega\setminus L$.
A new relation inside $A$ requires a comparison with both ends in $A$, since
$A$ is an initial segment of that order, and such a comparison removes at
most one maximal element of $A$. Dually, only comparisons inside $M$ remove
minimal elements of $M$, at most one each. At the end $e$ must be the only
maximal element of $A$ and the only minimal element of $M$, which needs at
least $a_s$ and $b_s$ comparisons respectively. As $A\cap M=\{e\}$, no
comparison is counted twice.
\end{proof}

\begin{corollary}\label{cor:combined}
$V_i(P)\ge\mathrm{LB}(P)=\max\{\lceil\log_2\Phi^*(P)\rceil,\KEY(P)\}$.
Since $2^t\ge t+1$ gives $\FG,\FGd\ge|C|$, this is never weaker than the
compatible-solution bound or either one-sided weight bound.
\end{corollary}

Within the family $\Psi_{F,G}=|C|+\sum_L F(b(L))+\sum_A G(a(A))$ with
$F(0)=G(0)=0$, the choice $F(t)=G(t)=2^t-t-1$ gives $\Phi^*$. It is the
largest valid member:

\begin{proposition}\label{prop:phiopt}
If $\Psi_{F,G}(E_n,1)\le2^{V_1(n)}$ and $\Psi_{F,G}(E_n,n)\le2^{V_n(n)}$ for
all $n\ge1$, then $F(t),G(t)\le2^t-t-1$ for all $t\ge0$, and hence
$\Psi_{F,G}\le\Phi^*$ on every problem.
\end{proposition}

\begin{proof}
For the minimum of $E_n$ there is one lower set with $b=n-1$ and $n$
singleton initial sets with $a=0$, so $n+F(n-1)\le2^{n-1}$. The maximum
gives the bound for $G$. Substituting termwise gives $\Psi_{F,G}\le\Phi^*$.
\end{proof}

\paragraph{Formal verification.}
A second Lean~4 development, based on Mathlib, formalizes
Lemma~\ref{lem:world}, Theorems~\ref{thm:phi} and~\ref{thm:key},
Corollary~\ref{cor:combined} and Proposition~\ref{prop:phiopt} for every
finite strict partial order and every rank $1\le i\le n$, with $C(P)$
defined through total orders extending $P$ as above. Its decision trees may
compare any two of the original elements, even when their order is known; a
tree is correct if it returns the $i$-th smallest element for every total
order extending $P$, and comparisons are counted along the path of each such
input. Theorem~\ref{thm:phi} is proved for the number of leaves and as
$\Phi^*(P)\le2^d$ for every correct tree making at most $d$ comparisons on
every input, with $\Phi^*$ computed in the integers. The proofs use only the
standard axioms and replay in \texttt{leanchecker}; the trusted base now
includes the Mathlib definitions used in the statements. The evaluation and
Remark~\ref{rem:leafcount} rest on computation and are not formalized.

\paragraph{Evaluation.}
At the empty order, $\lceil\log_2\Phi^*\rceil$ improves the one-sided
bound by one comparison for fifteen pairs $(n,i)$ with $n\le16$ and
$i\le\lceil n/2\rceil$, but gives
only 23 for $(16,8)$; it does not reproduce the lower bound 28. The local
inequality of Theorem~\ref{thm:phi} holds on all 149,598 cases with at most
7 elements (isomorphism classes, unsolved ranks, all incomparable pairs).
An independent exhaustive check over all partial orders with at most 9
elements, up to isomorphism, found no violation of
Theorems~\ref{thm:phi} and~\ref{thm:key} against exact values. On a corpus
of 1,663,508 exactly solved reduced problems with at most 12 elements,
gathered by our exact solver during its searches, the average gap to the
exact value is 2.62 for $\lceil\log_2|C|\rceil$ and 1.13 for $\mathrm{LB}$,
and $\mathrm{LB}$ is exact on 10\% of them.

\begin{remark}[A limit of leaf counting]\label{rem:leafcount}
No lower bound on the number of leaves of every correct tree can prove
$V_8(16)\ge29$ by taking its ceiling logarithm. The 252-function
29-comparison certificate unfolds into a tree with 248,381,440 syntactic
leaves, fewer than $2^{28}$, computed by two independent counts. Any bound
valid for all correct trees is at most this number, so its ceiling
logarithm is at most 28. A proof of 29, if true, must use depth or path
information.
\end{remark}

\section{How the algorithms were found}\label{sec:search}

The search is not part of the proofs; it only had to produce certificates.
Three changes to the forward search of~\cite{sea25} made this practical.

\paragraph{Fixed budget, first success.}
That search computes optima by trying
budgets upward from a lower bound and keeps looking for better comparisons
after finding one that works. We instead fix a budget $k$ and accept the
first comparison whose two outcomes are both solvable within $k-1$. One or
two comparisons of slack make construction much easier: for $(14,7)$, with
optimum 25, the time fell from 659\,s at the optimum to 119\,s and 17.6\,s
at one and two above it.

\paragraph{Information-greedy ordering.}
Candidate comparisons are tried in
increasing order of the larger compatible-solution count of their two
outcomes. On two benchmarks this reduced the searched nodes by factors of
about 20 and 100, and at the optimum the time for $(14,7)$ fell from
659\,s to 3.0\,s. With this ordering, budget 29 for $(16,8)$ was solved in
about two minutes and budget 29 for $(16,7)$ in 17 minutes, but budget 28
for $(16,7)$ remained unsolved after about 50 minutes. The bound
$\mathrm{LB}$ of Corollary~\ref{cor:combined} serves as the pruning rule.

\paragraph{Proof-number search.}
Better ordering alone did not reach the tightest budgets. Even an ordering
by $\Phi^*$, which put an optimal comparison first in about 96\% of
exactly solved problems with 7 to 10 elements, did not find a
26-comparison algorithm for $(15,7)$ by depth-first search. This suggested
that the difficulty lay in revising early choices rather than in making
them: depth-first search can spend very long refuting one early wrong
choice before trying another. Construction is an AND/OR
problem: one comparison must succeed at each state, and both of its
outcomes must succeed. Depth-first proof-number search
(df-pn)~\cite{allis,nagai,kishimoto} redirects effort to cheaper
alternatives before a choice is fully refuted. We initialize unexplored
nodes from a cost predictor fitted to the corpus of solved small problems,
shifted pessimistically, which mattered by factors of 10--100. Table~\ref{tab:dfpn}
compares the two searches; df-pn found the 28-comparison algorithm for
$(16,7)$ in five seconds.

\begin{table}[htbp]
\centering
\small
\begin{tabular}{lrrr}
\toprule
instance & df-pn & depth-first & functions \\
\midrule
$(14,7)$ at 25 & 0.55\,s & 3.0\,s & 412 \\
$(15,7)$ at 26 & 79\,s & not found (35\,min) & 1614 \\
$(16,7)$ at 28 & \textbf{5.0\,s} & not found (50\,min) & 329 \\
$(16,8)$ at 29 & 320\,s & 2\,min & 2937 \\
\bottomrule
\end{tabular}
\caption{Time to construct an algorithm at the stated budget (single thread).}
\label{tab:dfpn}
\end{table}

For 17 elements we widened the solver's representation to at most 19
elements; it reproduces earlier certificates byte for byte. Starting from
the upper bounds of~\cite{gkp}, df-pn reached the bounds of
Theorem~\ref{thm:main17} within 89\,s to 63\,min per rank (Table~\ref{tab:certs});
for rank 5 this meets the lower bound.

For $V_8(16)$ at 28, two df-pn runs of 6,600\,s and a six-hour portfolio
ended undecided. Refutations one below the optimum grew by factors of about
100--300 per added element on the instances we measured, so an exhaustive
refutation at 16 elements was out of reach on our desktop.

\section{Conclusion and open problems}

The previous best upper bounds for these entries came from a search
restricted to a few heuristically chosen comparisons per node~\cite{gkp};
the exhaustive search of~\cite{sea25} aims at optimal values, which were
out of reach for them. Our results indicate that much of the remaining gap
was a matter of search strategy. Asking only for some algorithm within a
fixed budget avoids computing optimal values of subproblems, and ordering
comparisons by the remaining uncertainty makes good choices likely. Near
the optimum this was not enough: the difficulty appeared to lie in revising
early choices, and proof-number search, which can turn to an alternative
before the current choice is fully refuted, overcame it. On some
instances, such as the median of 16 at budget 29, depth-first search was
faster, so the two strategies are complementary. Since every bound is
checked independently of the search, none of these heuristics needs to be
correct for the results to hold.

The same experience shows where the difficulty now lies. Constructing an
algorithm became cheap, but proving that none exists within a budget
remained expensive: refutations one below the optimum grew by factors of
about 100--300 per added element in our measurements, and the
28-comparison question for the median of 16 stayed undecided. For the open
entries, progress depends mainly on lower bounds.

\paragraph{Open problems.}
Is $V_8(16)$ equal to 28 or 29? A 28-comparison certificate would settle it
constructively; a proof of 29 must exclude all such algorithms and, by
Remark~\ref{rem:leafcount}, needs more than a total leaf count. An analytic
proof of $V_7(16)\ge28$ would remove the dependence on a large computation.
The gaps for $n=17$ at ranks 6--9 remain open. Finally, can the runner-up
opening be extended into a family of good selection algorithms, for example
for central ranks of $2^k$ elements?

\section*{Data availability}
Code, certificates, and Lean proofs are archived on Zenodo as release
v1.0.0~\cite{cremer-artifacts}. The development repository is available at
\url{https://github.com/JCremerr/selection-7-of-16}.

\clearpage
\section*{Use of AI}
This work was inspired by recent AI-assisted results in
combinatorics~\cite{cremer-ramsey} and was carried out with AI coding
assistants. The author chose the problem, directed the searches, and
developed the text with the assistants in successive revisions. The
assistants implemented the software and produced first drafts of the
proofs. The author is responsible for the content.

\section*{Acknowledgments}
We thank J.~D\"orrer, K.~Gendle, J.~Betz, J.~von Smercek, A.~Steding and
F.~Stober for publishing their search code and optimal algorithms.

\end{document}

%% file: fig_v7_structure.tex
\begin{figure}[!t]
\centering
\begin{tikzpicture}[x=1cm,y=1cm,
    el/.style={circle,fill=black,inner sep=0pt,minimum size=4pt},
    tr/.style={line width=0.5pt},
    ru/.style={line width=1.6pt,blue!70!black},
    lab/.style={font=\scriptsize,inner sep=1pt},
    box/.style={draw,rounded corners=2pt,font=\scriptsize,align=center,inner sep=2.5pt},
    cmp/.style={font=\scriptsize,inner sep=2.5pt},
    stage/.style={font=\footnotesize,align=center,anchor=north}]

\node[font=\small\bfseries] at (-0.95,9.15) {(a)};
\foreach \i in {0,...,3} { \node[el] at (-0.6+0.25*\i,8.3) {}; }
\node[stage] at (-0.22,7.72) {16 elements\\[-1pt]28 left};
\draw[->] (0.55,8.3) -- node[above,font=\scriptsize]{8 cmp.} (1.75,8.3);
\node[el] (p1) at (2.2,8.0) {}; \node[el] (p2) at (2.2,8.6) {};
\draw[tr] (p1) -- (p2);
\node[stage] at (2.2,7.72) {8 pairs\\[-1pt]20 left};
\draw[->] (2.65,8.3) -- node[above,font=\scriptsize]{4 cmp.} (3.85,8.3);
\node[el] (tw) at (4.6,7.95) {}; \node[el] (tb) at (4.3,8.4) {};
\node[el] (tr) at (4.9,8.4) {};  \node[el] (td) at (4.9,8.85) {};
\draw[tr] (tw)--(tb) (tw)--(tr) (tr)--(td);
\node[lab,left=2pt]  at (tw) {$w$}; \node[lab,left=2pt] at (tb) {$b$};
\node[lab,right=2pt] at (tr) {$r$}; \node[lab,right=2pt] at (td) {$d$};
\node[stage] at (4.6,7.72) {4 groups $T_4$\\[-1pt]16 left};
\draw[->] (5.45,8.3) -- node[above,font=\scriptsize]{2 cmp.} (6.65,8.3);
\begin{scope}[shift={(7.0,7.85)}]
  \node[el] (aw) at (0.3,0) {};   \node[el] (ar) at (0,0.45) {};
  \node[el] (ab) at (0.6,0.45) {}; \node[el] (ad) at (0,0.9) {};
  \node[el] (bw) at (1.3,0.45) {}; \node[el] (br) at (1.0,0.9) {};
  \node[el] (bb) at (1.6,0.9) {};  \node[el] (bd) at (1.0,1.35) {};
  \draw[tr] (aw)--(ar) (aw)--(ab) (ar)--(ad) (bw)--(br) (bw)--(bb) (br)--(bd);
  \draw[ru] (ar)--(br);
\end{scope}
\node[stage] at (7.8,7.72) {2 copies of $R_8$\\[-1pt]14 left};
\draw[->] (8.95,8.3) -- node[above,font=\scriptsize]{1 cmp.} (10.15,8.3);
\node[font=\footnotesize,align=center,anchor=west] at (10.3,8.3) {$R_{16}$, see (b)\\[-1pt]13 left};

\node[font=\small\bfseries] at (-0.95,6.3) {(b)};
\begin{scope}[shift={(0.2,2.2)},scale=0.95]
  \node[el] (wB) at (0,0.9) {};   \node[el] (bB) at (-0.5,1.8) {};
  \node[el] (rB) at (0.5,1.8) {}; \node[el] (dB) at (0.5,2.7) {};
  \node[el] (wA) at (1.8,0) {};   \node[el] (rA) at (1.3,0.9) {};
  \node[el] (bA) at (2.3,0.9) {}; \node[el] (dA) at (1.3,1.8) {};
  \node[el] (wC) at (3.6,0.9) {}; \node[el] (rC) at (3.1,1.8) {};
  \node[el] (bC) at (4.1,1.8) {}; \node[el] (dC) at (3.1,2.7) {};
  \node[el] (wD) at (5.4,1.8) {}; \node[el] (rD) at (4.9,2.7) {};
  \node[el] (bD) at (5.9,2.7) {}; \node[el] (dD) at (4.9,3.6) {};
  \foreach \t in {A,B,C,D} {
    \draw[tr] (w\t)--(b\t) (w\t)--(r\t) (r\t)--(d\t);
  }
  \draw[ru] (rA)--(rB) (rA)--(rC) (rC)--(rD);
  \foreach \t in {A,B,C,D} {
    \node[lab,below=2pt] at (w\t) {$w_{\t}$};
    \node[lab,above=2pt] at (b\t) {$b_{\t}$};
    \node[lab,above=2pt] at (d\t) {$d_{\t}$};
  }
  \node[lab,below left=1.5pt]  at (rA) {$r_A$};
  \node[lab,left=2pt]          at (rB) {$r_B$};
  \node[lab,above left=1.5pt]  at (rC) {$r_C$};
  \node[lab,above left=1.5pt]  at (rD) {$r_D$};
\end{scope}

\node[font=\small\bfseries] at (6.6,6.3) {(c)};
\node[box] (n0) at (10.6,6.0) {$R_{16}$: $(16,7)$, 13 left};
\node[box] (n1) at (8.45,4.45) {$(15,6)$, 12 left\\ drop $w_A$};
\node[box] (n2) at (12.75,4.45) {$(16,7)$, 12 left};
\node[box] (n11) at (7.35,2.45) {$(13,6)$, 11 left\\ drop $r_D,d_D$};
\node[box] (n12) at (9.55,2.45) {$(15,6)$, 11 left};
\node[box] (n21) at (11.75,2.45) {$(15,7)$, 11 left\\ drop $d_D$};
\node[box] (n22) at (13.95,2.45) {$(16,7)$, 11 left};
\draw[->] (n0) -- node[cmp,above left]{$r_B<b_C$} (n1);
\draw[->] (n0) -- node[cmp,above right]{$b_C<r_B$} (n2);
\draw[->] (n1) -- node[cmp,left]{$b_B<r_C$} (n11);
\draw[->] (n1) -- node[cmp,right]{$r_C<b_B$} (n12);
\draw[->] (n2) -- node[cmp,left]{$b_A<r_C$} (n21);
\draw[->] (n2) -- node[cmp,right]{$r_C<b_A$} (n22);
\end{tikzpicture}
\caption{The opening of the 28-comparison algorithm for the 7th smallest
of 16. (a)~Eight pair comparisons, four group finals, and three comparisons
among the runner-ups use 15 comparisons, leaving 13. Both outcomes at each
step give the same known order after renaming. (b)~The resulting order
$R_{16}$. Relations run upward from smaller to larger elements: an upward
path gives a known inequality. Edges implied by such paths are omitted.
In group $X$, $w_X$ beat $b_X$
and then $r_X$, while $r_X$ beat $d_X$. Blue thick edges compare runner-ups,
with $r_A<r_B$ and $r_A<r_C<r_D$. (c)~The next two comparisons. Each box
shows (remaining elements, requested rank) and comparisons left; ``drop''
names newly discarded elements. The left branch is worked through in the
text.}
\label{fig:v7structure}
\end{figure}
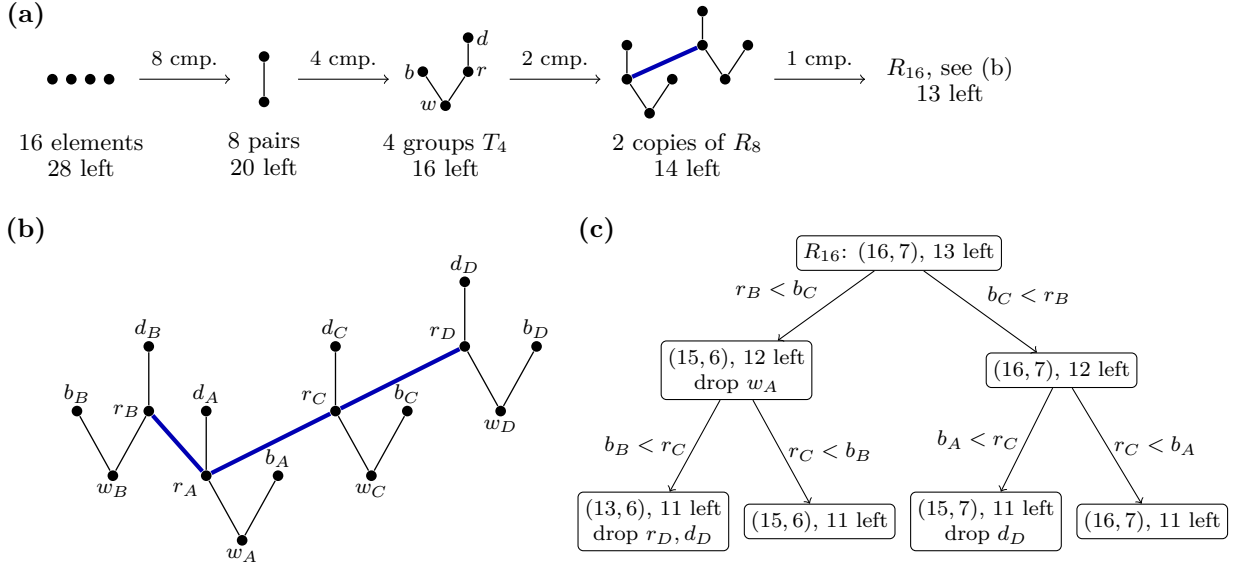